\documentclass[10pt, doublecolumn]{IEEEtran}
\usepackage{epsfig,latexsym}
\usepackage{float}
\usepackage{indentfirst}
\usepackage{amsmath}
\usepackage{bm}
\usepackage{amssymb}
\usepackage{times}

\usepackage{algorithm}
\usepackage[noend]{algpseudocode}
\usepackage{subfigure}
\usepackage{psfrag}
\usepackage{hyperref}
\usepackage{cite}
\usepackage{lastpage}
\usepackage{fancyhdr}
\usepackage{color}
 \usepackage{amsthm}
\usepackage{bigints}
\newtheorem{Lemma}{Lemma}
\newtheorem{Corollary}{Corollary}
\newtheorem{lemma}[Lemma]{$\mathbf{Lemma}$}

\newtheorem{corollary}[Corollary]{$\mathbf{Corollary}$}

\newcounter{problem}
\newcounter{save@equation}
\newcounter{save@problem}
\makeatletter

\begin{document}
\title{\vspace{-1em} \huge{ A New Design of   CR-NOMA and Its Application \\ to AoI Reduction   }}

\author{ Zhiguo Ding,  Octavia A. Dobre,   Pingzhi Fan,     and H. Vincent Poor   \thanks{ 
  
\vspace{-2em}

  Z. Ding is with Department of Electrical Engineering and Computer Science, Khalifa University, Abu Dhabi, UAE, and Department of Electrical and Electronic Engineering, University of Manchester, Manchester, UK. (email: \href{mailto:zhiguo.ding@manchester.ac.uk}{zhiguo.ding@manchester.ac.uk}). O. A. Dobre is with the Faculty of
Engineering and Applied Science, Memorial University, St. John’s, NL A1B
3X5, Canada (e-mail: \href{odobre@mun.ca}{odobre@mun.ca}). P. Fan is with the Institute of Information Coding  \& Transmission Key Lab of Sichuan Province, Southwest Jiaotong University, Chengdu, China (email: \href{mailto:p.fan@ieee.org }{p.fan@ieee.org}).
H. V. Poor is  with the Department of
Electrical and Computer Engineering, Princeton University, Princeton, NJ 08544,
USA, (email:  \href{mailto:poor@princeton.edu}{poor@princeton.edu}).

  }\vspace{-2em}}
 \maketitle

\vspace{-1em}
\begin{abstract}
The aim of this letter is to  develop a new design of cognitive radio  inspired non-orthogonal multiple access (CR-NOMA), which ensures that  multiple new users can be supported without causing disruption to the legacy network.   Analytical results are developed to characterize the statistical properties of the number of   supported new users. The developed CR-NOMA scheme is compatible  to various communication networks, because it can be implemented as a  simple  add-on. An example of using  CR-NOMA as an add-on in time-division multiple access   networks for age of information (AoI) reduction is illustrated,  where analytical results are developed to demonstrate the superior performance gain of using CR-NOMA. 
\end{abstract}\vspace{-0.1em}

\vspace{-1em} 

\section{Introduction}
One of the key features of non-orthogonal multiple access (NOMA) is its superior compatibility with other communication techniques, because  NOMA can be   implemented as a type of simple add-ons in other legacy communication networks \cite{9711564}.  Take cognitive radio inspired NOMA (CR-NOMA) as an example \cite{Zhiguo_CRconoma}. One   secondary user is admitted into the bandwidth resource block which would be solely occupied by a primary user in orthogonal multiple access (OMA). Note that the secondary user is served without causing any disruption to the primary user, and  hence serving the secondary user is transparent to the primary user. In other words, by applying CR-NOMA as a simple   add-on,     the user connectivity and the overall system throughput of many communication networks can be improved in a spectrally efficient manner \cite{10029912, 9802906,9606870}. 
Most existing works on CR-NOMA have been focused on the case with a single secondary user, which motivates this letter.

In particular, the aim of this letter is to propose a new design of CR-NOMA in order to support multiple secondary users. The key challenges for design multi-user CR-NOMA are how  to determine the number of supported secondary users and how to  decide on their transmit powers.  The fact that the two challenges are coupled with the design of successive interference cancellation (SIC) makes the design of multi-user CR-NOMA more difficult. In this letter, the key idea of random access NOMA is applied to CR-NOMA in order to dynamically determine the number of supported secondary users and also simplify the problem of transmit power allocation \cite{jsacnoma10}.  Analytical results are developed to characterize the statistical property of the  number of supported secondary users, as well as the secondary users' sum-rate. In order to demonstrate the compatibility of CR-NOMA,   an example of using  CR-NOMA as an add-on in time division multiple access (TDMA)   networks for age of information (AoI) reduction is illustrated,  where analytical results are developed to demonstrate the superior performance gain of using the developed new design of CR-NOMA \cite{9380899}.

\vspace{-1.5em}
\section{A New Design of CR-NOMA}
Assume that there exists an OMA-based legacy uplink  network, where a user, termed the primary user and denoted by ${\rm U}_0$, has been  scheduled   on an orthogonal resource block, e.g., a time slot.  The key idea of CR-NOMA is to support   additional users, termed secondary users and denoted by ${\rm U}_i$, $i=1, 2, \cdots$,  without degrading ${\rm U}_0$'s quality-of-service (QoS) requirements.  To facilitate the performance analysis, the   users' channels   are assumed to be  complex Gaussian independent and identically distributed (i.i.d.) with zero mean and unit variance. The proposed   CR-NOMA scheme is described as follows. 

\vspace{-1em}
\subsection{Description of Multi-User CR-NOMA}
\subsubsection{Determine  ${\rm U}_0$'s tolerable interference}
Similar to   conventional  CR-NOMA \cite{Zhiguo_CRconoma}, which can support a single secondary user only, for the new CR-NOMA design, the base station also decodes ${\rm U}_0$'s signal first, which means that the admitted secondary users' signals become the interference to the primary user. Therefore, it is necessary to  first   determine how much interference can be tolerated by  the primary user ${\rm U}_0$, which directly affects how many secondary users can be supported, as shown in the following subsection.  
 In particular, denote the maximal interference tolerable to ${\rm U}_0$    by $I$, and  hence $I$ needs to satisfy the following constraint:  
\begin{align}\label{ccc}
\log\left(1+\frac{P |h_0|^2}{1+I}\right)= R\Rightarrow   I= \frac{P |h_0|^2}{\epsilon}-1,
\end{align}
where the noise power is assumed to be normalized to be one,  $R$ denotes ${\rm U}_0$'s target data rate, $h_0$ denotes ${\rm U}_0$'s channel gain,  $\epsilon=2^R-1$, and   $P$ denotes ${\rm U}_0$'s transmit signal-to-noise (SNR).   For illustration purposes,  it is assumed that the primary and secondary users  have the same target data rate, and all the users have the same transmit power budget, i.e., $P$.

The key challenge to design multi-user CR-NOMA is due to the following two coupled issues, namely how many secondary users can be supported and how to decide on their transmit powers, which are   further entangled  with the design  of the SIC decoding order of the secondary users' signals at the base station. The key step for the new design of CR-NOMA is to integrate another form of NOMA, termed random access NOMA \cite{jsacnoma10}, into CR-NOMA, which can effectively decouple the two aforementioned issues, as shown in the following.

\subsubsection{Pre-configure the NOMA Receive SNR levels}
Based on random access NOMA, the receive SNR levels  for the secondary users' signals at the base station can be preconfigured. In particular, prior to transmission, the base station defines a series of receive SNR levels, denoted by $P_k$, $k=1, 2, \cdots$, where $P_k<P_j$ for $k<j$. The lowest SNR level, $P_1$, is to be used by the secondary user whose signal is decoded at the last stage of SIC. In order to support the user's target data rate,  $P_1$ needs to satisfy the following constraint:
\begin{align}\label{siclast}
\log\left(1+ P_1 \right)= R
\Rightarrow
P_1  = \epsilon .
\end{align}
The SNR level, $P_2$, is to be used by the secondary user whose signal is decoded in the second-to-last SIC stage, and hence it needs to satisfy the following constraint: $\log\left(1+\frac{P_2}{1+ P_1}\right)= R$. In general, the choice of the SNR level, $P_k$, can be obtained from $P_i$, $1\leq i \leq k-1$, as follows: 
\begin{align}\label{sick}
\log\left(1+\frac{P_k}{1+ \sum^{k-1}_{i=1}P_i}\right)= R \Rightarrow P_k = \epsilon\left(1+ \sum^{k-1}_{i=1}P_i \right).
 \end{align}

\subsubsection{Determine $K$ supportable   by ${\rm U}_0$}   Note that only a limited number of SNR levels  can be supported by ${\rm U}_0$, because the total interference from the secondary users needs to be capped by $I$, in order to avoid any performance degradation to ${\rm U}_0$. In particular, consider that the  first $K$ SNR levels are to be used, which means  that $K$ needs to satisfy  the constraint: 
\begin{align} \label{criterion}
\sum_{k=1}^{K}P_k \leq I,
\end{align} 
if $I>0$, otherwise $K=0$, i.e., no secondary user can be supported. Note that $K$ is an important system parameter since it is directly related to how many secondary users can be supported by ${\rm U}_0$.  The following lemma shows the statistical properties of $K$.

 \begin{lemma}\label{lemma1}
 The probability for the number of SNR levels   supported by   ${\rm U}_0$ is given by 
  \begin{align}
 {\rm P}(K=n)=&    
 e^{- \frac{\epsilon+\epsilon \eta_n}{P} }-e^{- \frac{\epsilon+\epsilon \eta_{n+1} }{P}}  ,
 \end{align}
 for $n\geq 1$, and $ {\rm P}(K=0)= 1- e^{-\frac{\epsilon^2+\epsilon }{P}
} $, where   $\eta_n=\sum_{k=1}^{n}P_k$. The mean of $K$ is given by $\mathcal{E}\{K\}=\sum^{\infty}_{n=0}n {\rm P}(K=n)$, and  $\mathcal{E}\{K\}\rightarrow \infty$ for $P\rightarrow \infty$. 
 \end{lemma}
 \begin{proof}
 See Appendix \ref{proof1}.
 \end{proof}
Lemma \ref{lemma1} indicates that the number of secondary users which can be supported by ${\rm U}_0$ is unbounded at high SNR,  which is an important property to support massive connectivity in the next-generation wireless network.
   
\subsubsection{Schedule $K$ users} Once the $K$ SNR levels supported by ${\rm U}_0$ are known, the next task is to schedule secondary users for using these preconfigured SNR levels. 

Without loss of generality, assume that there are $M$ secondary users to be scheduled. 
A low-complexity   scheduling approach is to randomly select $K$ users from the $M$  candidates, which does  not require the users' channel state information (CSI).  An alternative greedy   approach is to schedule the user with the $k$-th weakest channel condition on $P_k$, which is motivated by the following observation.  Supposing  that ${\rm U}_i$ chooses $P_k$, its transmit SNR needs to be $\frac{P_k}{|h_i|^2}$, where $h_i$ denotes ${\rm U}_i$'s channel gain. If  $\frac{P_k}{|h_i|^2}>P$, $P_k$  is not feasible to  ${\rm U}_i$ due to the user's transmit power budget. Therefore, it is important to assign a user with strong channel conditions to a high SNR level; however, this greedy approach requires the base station to have access to all the users' CSI.

For both   scheduling approaches,   a user simply keeps silent if it finds its assigned   SNR level not feasible, which is important to avoid the situation where one user's transmission failure causes  a catastrophe to the whole network. In particular, for conventional uplink NOMA, if the base station fails to detect one user's signal,   the remaining SIC stages to detect other users' signals have to be abandoned.   For the new CR-NOMA scheme, for  a user which cannot afford the assigned SNR level, its  transmission failure will not cause any performance degradation to the other users. 
\vspace{-1em}
\subsection{Performance Analysis for CR-NOMA}
Due to space limitations, hence we focus on the performance analysis for   the random scheduling approach. For simplicity purposes, assume that ${\rm U}_k$ is assigned with the SNR level $P_k$. The secondary users' outage sum-rate achieved by CR-NOMA can be expressed as follows: 
\begin{align}
R_{\rm sum} =&
 \sum^{M}_{n=1}{\rm P}\left(  K=n\right)\sum_{k=1}^{n}  {\rm P}\left(\left. E_k\right | K=n\right)kR  ,
\end{align}
where $E_k$ denotes an event that there are $k$ users successfully delivering their updates to the base station.

For conventional CR-NOMA, a failure in one SIC stage stops    the remaining SIC stages, which makes the performance analysis very challenging.   For  the new CR-NOMA, because   a user whose assigned SNR level is infeasible keeps silent, this user's failure will   not cause any disruption  to the other users. Therefore, the outage sum-rate  can be simplified as follows:
\begin{align}
R_{\rm sum} =& 
 \sum^{M}_{n=1}{\rm P}\left(  K=n\right)R \sum_{k=1}^{n}{\rm P}\left(\left. \frac{P_k}{|h_k|^2}\leq P\right | K=n\right) .
\end{align}
By applying Lemma \ref{lemma1} and using the fact that the number of secondary users is capped by $M$, the secondary users' sum-rate can be obtained as follows.
 \begin{corollary}\label{corollary1}
 With random scheduling, the sum-rate of the secondary users is given by
 \begin{align}
R_{\rm sum} =& 
 \sum^{M-1}_{n=1}\left(
  e^{- \frac{\epsilon+\epsilon \eta_n}{P} }-e^{- \frac{\epsilon+\epsilon \eta_{n+1} }{P}} 
 \right)
 R \sum_{k=1}^{n} e^{-\frac{P_k}{P}}\\\nonumber
 &+ \sum^{\infty}_{n=M}\left(
  e^{- \frac{\epsilon+\epsilon \eta_n}{P} }-e^{- \frac{\epsilon+\epsilon \eta_{n+1} }{P}} 
 \right)
 R \sum_{k=1}^{M} e^{-\frac{P_k}{P}}.
\end{align}
 \end{corollary}
 
 \section{Application of CR-NOMA  for AoI Reduction}
 The new design of CR-NOMA can be straightforwardly applied to TDMA networks as an add-on  for AoI reduction. In particular, consider that there exists a TDMA-based legacy network, where each time frame consists of $N$ time slot. To be consistent  to the notations in the previous section, assume that there are $(M+1)N$ users which want to send their updates to the base station. With round-robin TDMA, a super-frame which consists of $M$   frames is  needed to serve all the users. Denote by ${\rm U}_n^k$  a user scheduled in the $n$-th time slot of the $k$-frame of a super-frame.  Assume that the users' channel gains in different time slots are i.i.d. complex Gaussian variables with zero mean and unit variance. 
 
Without loss of generality, we focus on the AoI of ${\rm U}_1^1$, where  the user's  average AoI   is defined as follows: 
\begin{align}
\bar{\Delta} =  \underset{\bar{T}\rightarrow\infty}{\lim} \frac{1}{\bar{T}}\int^{\bar{T}}_{0}\Delta(t) dt, 
\end{align} 
where $\Delta(t)  $ denotes the user's instantaneous AoI at time $t$. By following the steps similar to those in \cite{crnomaaoi}, the AoI achieved by TDMA can be straightforwardly obtained as follows: 
  \begin{align}
 \bar{\Delta}^{T} =   T+(M+1)NT\left(2e^{\frac{\epsilon}{P}}-1\right),
 \end{align}
 where $T$ denotes the duration of each time slot, and the generate-at-will model is considered. 
 
 In order to apply CR-NOMA,   the $(M+1)$ users which are scheduled  at the same time slot of the $(M+1)$ frames in one super-frame, e.g., ${\rm U}_n^1$, $\cdots$, ${\rm U}_n^{M+1}$,  are grouped for the implementation of NOMA.  
 
The AoI of ${\rm U}_1^1$  is focused on, without loss of generality. In the first time slot of the first frame of each super-frame, ${\rm U}_1^1$ is scheduled to transmit, in the same manner as in OMA.  Following the descriptions in the previous section, ${\rm U}_1^1$ can be viewed as the primary user, ${\rm U}_0$, where ${\rm U}_1^2$, $\cdots$, ${\rm U}_1^{M+1}$ can be viewed as the secondary users. Note that not all secondary users can be supported by ${\rm U}_1^1$. Again by using the steps in the previous section, $K$ SNR levels which can be supported by ${\rm U}_1^1$,   $P_1$, $\cdots$, $P_{K}$, are assigned to    ${\rm U}_1^2$,  $\cdots$,  ${\rm U}_1^{K+1}$, respectively. If a user is not assigned a SNR level or finds that its assigned SNR level is infeasible due to its power budget, it keeps silent. An illustration of SNR level allocation for the case $K=M$ is shown in Table \ref{table1}.

\begin{table}\vspace{-2em}
  \centering
  \caption{Illustration of SNR Level Allocation in One Super-Frame.}
  \begin{tabular}{|c|c|c|c|c|c|}
\hline
&  frame $1$ &frame $2$ &$\cdots$&frame $M$&frame $M+1$ \\
    \hline
   $P$  &  ${\rm U}_1^1$	&${\rm U}_1^2$ &$\cdots$ &${\rm U}_1^{M}$ &${\rm U}_1^{M+1}$\\\hline
   $P_1$  &  ${\rm U}_1^2$	&${\rm U}_1^3$ &$\cdots$ &${\rm U}_1^{M+1}$ &${\rm U}_1^{1}$\\
    \hline
   $P_2$  &  ${\rm U}_1^3$	&${\rm U}_1^4$ &$\cdots$ &${\rm U}_1^1$ &${\rm U}_1^{2}$\\    \hline
   $\vdots$  &  $\vdots$ 	&$\vdots$  &$\cdots$ &$\vdots$ &$\vdots$ \\
       \hline
   $P_{M-1}$  &  ${\rm U}_1^M$	&${\rm U}_1^{M+1}$ &$\cdots$ &${\rm U}_1^{M-2}$ &${\rm U}_1^{M-1}$\\
       \hline
   $P_{M}$  &  ${\rm U}_1^{M+1}$	&${\rm U}_1^1$ &$\cdots$ &${\rm U}_1^{M-1}$ &${\rm U}_1^{M}$\\\hline
  \end{tabular}\vspace{1em}\label{table1}\vspace{-2em}
\end{table}

As shown in the table, in the  $m$-th frame, $m\geq 2$, ${\rm U}_1^m$ is treated as the primary user, and ${\rm U}_1^1$ is assigned with $P_{M-m+2}$, if this SNR level can be supported by ${\rm U}_1^m$. By using Lemma \ref{lemma1} and Table \ref{table1}, the following lemma can be obtained for the   AoI achieved by CR-NOMA.

\begin{lemma}\label{lemma2}
The AoI achieved by the proposed multi-user      CR-NOMA scheme is given by
 \begin{align}
 \bar{\Delta} = T +\frac{\Delta_Y}{\Delta_{Y^2}},
 \end{align}
where  
  $\Delta_Y$ is expressed as follows: 
 \begin{align}
\Delta_Y =& \sum^{\infty}_{q=0} \sum^{M+1}_{m=1}
 \sum^{M+1}_{n=1}\left[  (q+1)(M+1)+ n-m\right]NT\\\nonumber &\times \phi_m\phi_n  \psi_{m+1}^{M+1}\psi_{1}^{n-1}\left(1-e^{-\frac{\epsilon}{P}} \right)^q\prod_{p=2}^{M+1}\left(1- \phi_p\right) ^q\\\nonumber &+
  \sum^{M}_{m=1}
 \sum^{M+1}_{n=m+1} (n-m)NT \phi_m\phi_n \psi_{m+1}^{n-1} ,
 \end{align}
 \begin{align}\nonumber 
\Delta_{Y^2} =& \sum^{\infty}_{q=0} \sum^{M+1}_{m=1}
 \sum^{M+1}_{n=1}\left[  (q+1)(M+1)+ n-m \right]^2N^2T^2\\  &\times \phi_m\phi_n  \psi_{m+1}^{M+1}\psi_{1}^{n-1}\left(1-e^{-\frac{\epsilon}{P}} \right)^q\prod_{p=2}^{M+1}\left(1- \phi_p\right) ^q
 \nonumber\\  &+
  \sum^{M}_{m=1}
 \sum^{M+1}_{n=m+1} (n-m)^2N^2T^2  \phi_m\phi_n \psi_{m+1}^{n-1},
 \end{align}
$\psi_{m}^{n} = \prod_{p=m}^{n}\left(1-\phi_p\right)$,  $\phi_1=e^{-\frac{\epsilon}{P}}$ and, for $2\leq m\leq M+1$,
 \begin{align}
\phi_m
 =& \sum^{\infty}_{k=M-m+2}\left(  
 e^{- \frac{\epsilon+\epsilon \eta_k}{P} }-e^{- \frac{\epsilon+\epsilon \eta_{k+1} }{P}}\right)e^{-\frac{P_{M-m+2}}{P}}. 
 \end{align}
\end{lemma}
 \begin{proof}
 See Appendix \ref{proof2}.
 \end{proof}

 \section{Simulation Results}
 In this section, simulation results are presented to demonstrate the performance of the developed CR-NOMA scheme and its impact on the AoI.
 
 In Fig. \ref{fig1}, the characteristics  of CR-NOMA are focused on first.  Recall that the key step of CR-NOMA is to determine $K$, which is an important parameter because  it directly affects how many secondary users can be supported.   Fig. \ref{fig1a} demonstrates that the number of SNR levels, or the number of supported secondary users, is increased by increasing the transmit SNR, which confirms the conclusion predicted in Lemma \ref{lemma1}. In Fig. \ref{fig1b}, the secondary users' outage sum-rate is used as the metric to evaluate the performance of CR-NOMA, where both the random and greedy scheduling schemes are considered. As can be seen from the figure,  CR-NOMA can support a significantly large sum-rate for the secondary users, and it is important to point out that this data rate gain is achieved without consuming extra bandwidth resources. Both Figs. \ref{fig1a}  and \ref{fig1b} also verify the accuracy of the analytical results developed in Lemma \ref{lemma1} and Corollary \ref{corollary1}. 
 
 In Fig. \ref{fig2}, the impact of CR-NOMA on the AoI reduction is studied. Both   Figs. \ref{fig2a} and   \ref{fig2b} show that  CR-NOMA can significantly reduce the AoI, compared to OMA. Furthermore, Fig. \ref{fig2} shows that the performance gain of CR-NOMA over OMA can be further improved by increasing the number of users. Therefore, the application  of CR-NOMA will be particularly important to the future wireless network with massive users. Comparing Fig. \ref{fig2a} to \ref{fig2b}, one can observe that the performance gain of CR-NOMA over OMA is larger in the case of a smaller target data rate. This is also a valuable property given the fact that for most applications of machine-type communications, the users are expected   to be served with small data rates. Fig. \ref{fig2} also demonstrates the accuracy of the analytical results developed in Lemma \ref{lemma2}.  
 
 \begin{figure}[t] \vspace{-2em}
\begin{center}
\subfigure[Mean of $K$ ($\mathcal{E}\{K\}$)]{\label{fig1a}\includegraphics[width=0.35\textwidth]{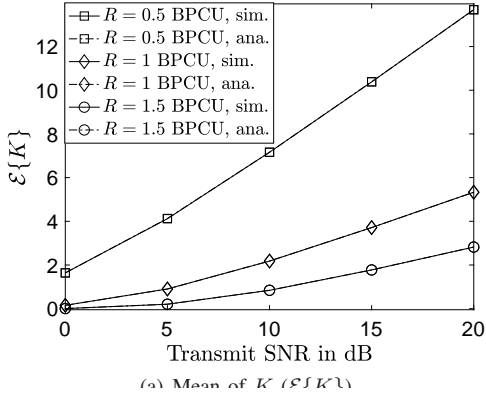}}\hspace{2em}
\subfigure[Secondary users' sum-rate, $M=8$]{\label{fig1b}\includegraphics[width=0.35\textwidth]{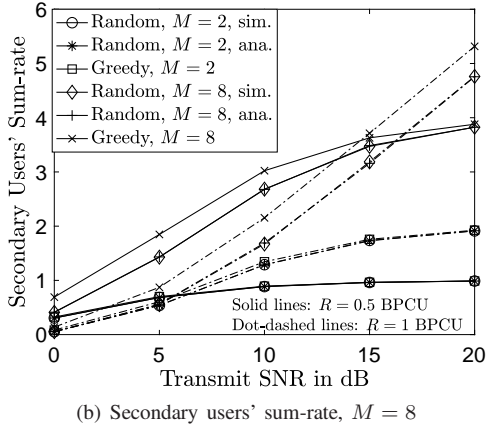}} \vspace{-1em}
\end{center}
\caption{ Impact of the proposed CR-NOMA transmission scheme.  BPCU denotes bits per channel use.   \vspace{-1em} }\label{fig1}\vspace{-1em}
\end{figure}
 
  \begin{figure}[t] \vspace{-2em}
\begin{center}
\subfigure[$R=0.5$ BPCU]{\label{fig2a}\includegraphics[width=0.35\textwidth]{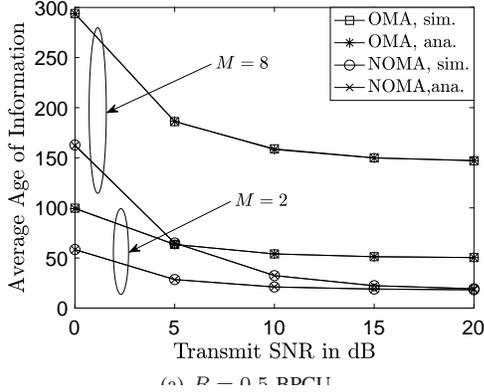}}\hspace{2em}
\subfigure[$R=1$ BPCU]{\label{fig2b}\includegraphics[width=0.35\textwidth]{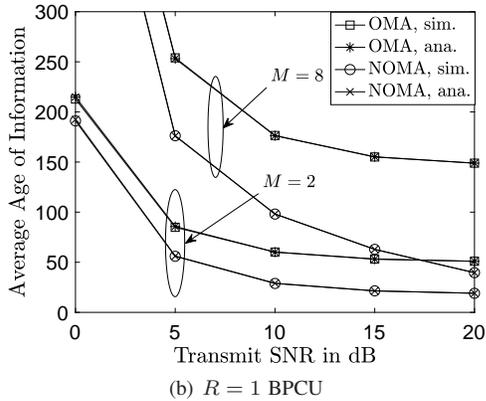}} \vspace{-1em}
\end{center}
\caption{ AoI achieved by the proposed CR-NOMA scheme. $T=2$ seconds and $N=8$.    \vspace{-1em} }\label{fig2}\vspace{-1em}
\end{figure}

\section{Conclusions}
In this paper,  a new form of CR-NOMA has been developed in order to support multiple secondary users. Analytical results have been  developed to characterize the statistical properties of the number of supported secondary users. In addition, an example of using  CR-NOMA as an add-on in TDMA   networks for AoI reduction was  illustrated,  where analytical results were developed to demonstrate the superior performance gain of using CR-NOMA. 
 
 \appendices 
 \section{Proof for Lemma \ref{lemma1}}\label{proof1}
The probability to support $K$ secondary users is given by
 \begin{align}
 {\rm P}(K=n)=& {\rm P}\left( \sum_{k=1}^{n}P_k \leq I, \sum_{k=1}^{n+1}P_k \geq  I, I\geq 0
\right) .
 \end{align}
  
 We note that   $P_k$ is not a function of $I$, and therefore $ {\rm P}(K=n)$ can be expressed as follows:
  \begin{align}
 {\rm P}(K=n)=&   {\rm P}\left(   \eta_n \leq I\leq  \eta_{n+1} 
\right),
 \end{align}
 for $n\geq 1$, 
 where recall that $\eta_n=\sum_{k=1}^{n}P_k$. 
The event that $ {\rm P}(K=0)$ is caused by  either $I<0$ or $I$ being too small to support $P_1$, i.e.,
  \begin{align}
 {\rm P}(K=0)=&   {\rm P}\left(  I<0 
\right)+  {\rm P}\left(0\leq I\leq \eta_1
\right) = {\rm P}\left(  I\leq \eta_1
\right) .
 \end{align}
 
By using \eqref{ccc},  the probability can be further rewritten as follows:
  \begin{align}
 {\rm P}(K=n)=&   
{\rm P}\left(  \eta_n  \leq   \frac{P |h_0|^2}{\epsilon}-1\leq  \eta_{n+1} 
\right)
\\\nonumber
=&  
 e^{- \frac{\epsilon+\epsilon \eta_n}{P} }-e^{- \frac{\epsilon+\epsilon \eta_{n+1} }{P}}  ,
 \end{align}
 where $n\geq 1$.
 
Furthermore, the probability  $ {\rm P}(K=0)$ can be obtained as follows:
 \begin{align}
 {\rm P}(K=0)=&  {\rm P}\left(  I \leq \epsilon 
\right)={\rm P}\left(    \frac{P |h_0|^2}{\epsilon}-1 \leq \epsilon 
\right)\\\nonumber
=&{\rm P}\left(   |h_0|^2 \leq \frac{\epsilon^2+\epsilon }{P}
\right)
=1- e^{-\frac{\epsilon^2+\epsilon }{P}
}.
 \end{align}
 
The mean of the number of secondary users which can be supported by ${\rm U}_0$ is given by 
  \begin{align}
\mathcal{E}\{K\}  = &\sum^{\infty}_{n=0} n{\rm P}(K=n) \\\nonumber
=&\sum^{\infty}_{n=1}  n\left(
 e^{- \frac{\epsilon+\epsilon \eta_n}{P} }-e^{- \frac{\epsilon+\epsilon \eta_{n+1} }{P}}\right) .
 \end{align}
 
By applying Markov's inequality \cite{Papoulisbook}, with a fixed $n$, a lower bound on $\mathcal{E}\{K\} $ can be obtained as follows:
 \begin{align}
 {\rm P}(K\geq n)\leq \frac{\mathcal{E}\{K\}}{n}\Rightarrow
\mathcal{E}\{K\}\geq   n{\rm P}(K\geq n).
 \end{align}
 This lower bound can be further rewritten  as follows:
  \begin{align}
\mathcal{E}\{K\}\geq  & n\left(1- {\rm P}(K\leq  n)\right)
\\\nonumber
=  & n\left(1- \sum^{n}_{i=0}{\rm P}(K=  i)\right).
 \end{align}
 It is important to note that $ {\rm P}(K=i)$, $i\geq 1$, can be approximated  as follows:
  \begin{align}
 {\rm P}(K=i)=&   
 e^{- \frac{\epsilon+\epsilon \eta_i}{P} }-e^{- \frac{\epsilon+\epsilon \eta_{i+1} }{P}} \\\nonumber \approx& \epsilon \frac{\eta_{i+1}-\eta_{i}}{P}\rightarrow 0,
 \end{align}
 where the approximation is obtained by assuming that $i$ is fixed and $P\rightarrow \infty$. The probability $ {\rm P}(K=0)$ also goes to zero for $P\rightarrow \infty$. 
 
 Therefore, at high SNR, the mean of $K$ can be lower bounded as follows:
   \begin{align}
\mathcal{E}\{K\} 
\geq  & n\left(1- \sum^{n}_{i=0}{\rm P}(K=  i)\right)\approx n,
 \end{align}
 which holds for any fixed $n$ with $P\rightarrow \infty$. In other words, $\mathcal{E}\{K\} $ is unbounded at high SNR.  The lemma is proved.
  
 \section{Proof for Lemma \ref{lemma2}} \label{proof2}
 First define the following event
 \begin{align}
 E_{jm} = \{\text{The user's $j$-th successful update happens} \\\text{in the $m$-th frame of a super-frame}\},
 \end{align}
 where $1\leq m\leq M+1$. 
 By using $E_{jm}$, we can define the duration between the $(j-1)$-th and $j$-th successful transmissions, denoted by $y_j$. In particular, conditioned on $ E_{(j-1)m}$, $ E_{jn}$, and $q$ super-frames being elapsed between the two adjacent successful transmissions, $y_j$ is given by 
 \begin{align}
 y_j = (M+1-m)NT+qNT+ nNT,
 \end{align}
 where $n>m$, if $q=0$. 
 
 By following the steps similar to those in \cite{crnomaaoi}, the average AoI can be obtained by using the statistics of $y_j$, i.e., 
 \begin{align}\label{exxx}
 \bar{\Delta} = T +\frac{\mathcal{E}\{Y^2\}}{\mathcal{E}\{Y\}},
 \end{align}
 where $\mathcal{E}\{Y^2\} = \underset{J\rightarrow \infty}{\lim} \sum^{J}_{j=1}y_j^2$ and $\mathcal{E}\{Y\} = \underset{J\rightarrow \infty}{\lim} \sum^{J}_{j=1}y_j$.
 
$\mathcal{E}\{Y\}$ can be expressed as follows:  
 {\small \begin{align}\nonumber
 \mathcal{E}\{Y\} =& \sum^{\infty}_{q=0} \sum^{M+1}_{m=1}
 \sum^{M+1}_{n=1}\left[ (M+1-m)NT+q(M+1)NT+ nNT \right]\\\label{xdfd} &\times {\rm P}( E_{(j-1)m},E_{jn}, x_j = q, E_{m+1}^{M+1},E_{1}^{n-1})\\\nonumber &+
  \sum^{M}_{m=1}
 \sum^{M+1}_{n=m+1} (n-m)NT  {\rm P}( E_{(j-1)m}, E_{jn}, E_{m+1}^{n-1}),
 \end{align}}
 $\hspace{-0.4em}$where $E_i^k$ denotes the event that the user fails to update the base station from the $i$-th frame to the $k$-th frame within one super-frame, and $x_j$ denotes the number of super-frames between the $(j-1)$-th and $j$-th   successful transmissions.  The last term in \eqref{xdfd} is for the event that the two adjacent successful updates   happen in one super-frame. 
 
Because   the users' channel gains are i.i.d., ${\rm P}( E_{(j-1)m})={\rm P}( E_{jm})$. Recall that $  E_{jm}$ is the event that ${\rm U}_1^1$ succeeds in the $m$-th frame.  According to the principle of CR-NOMA,  in the first frame, i.e., $m=1$, ${\rm U}_1^1$ acts as the primary user, and ${\rm P}( E_{j1})$ is simply given by ${\rm P}( E_{j1}) = e^{-\frac{\epsilon}{P}}$. In the $m$-th frame, $2\leq m \leq M+1$, ${\rm U}_1^1$ 
 acts as a secondary user and is  assigned with the SNR level, $P_{M-m+2}$, where ${\rm P}( E_{jm})$ can be calculated as follows:
 \begin{align}
 {\rm P}( E_{jm}) =& {\rm P}\left( E_{Pm},  \frac{P_{M-m+2}}{|h_1^m|^2}\leq P\right),
 \end{align}
  where $E_{Pm}$ denotes the event that $P_{M-m+2}$ can be supported by ${\rm U}_1^m$, and $h_1^m$ denotes  ${\rm U}_1^1$'s channel gain in the $m$-th frame. $ {\rm P}( E_{jm}) $ can be further evaluated as follows:
 \begin{align}
 {\rm P}( E_{jm}) =& {\rm P}\left( K_m\geq M-m+2,  \frac{P_{M-m+2}}{|h_1^m|^2}\leq P\right) 
  \\\nonumber
 =& \sum^{\infty}_{k=M-m+2}{\rm P}\left( K_m=k\right) {\rm P}\left(  \frac{P_{M-m+2}}{|h_1^m|^2}\leq P\right) ,
 \end{align}
where   $K_m$ denotes the number of SNR levels which can be supported by the primary user, ${\rm U}_1^m$. 
 
 By applying Lemma \ref{lemma1}, $ {\rm P}( E_{jm}) $ can be expressed as follows:  
  \begin{align}\label{eqx2}
 {\rm P}( E_{jm})  
 =& \sum^{\infty}_{k=M-m+2}\left(  
 e^{- \frac{\epsilon+\epsilon \eta_k}{P} }-e^{- \frac{\epsilon+\epsilon \eta_{k+1} }{P}}\right)e^{-\frac{P_{M-m+2}}{P}},
 \end{align}
 for $2\leq m \leq M+1$.  
 
The expression of $ {\rm P}( E_{jm})  $ can be used to find the probability of $x_j$ as follows:
  \begin{align}\nonumber
 {\rm P}( x_j=q) =&\left[\left(1-{\rm P}( E_{j1})  \right)\prod_{p=2}^{M+1}\left(1- {\rm P}( E_{jp})\right) \right]^q  \\ 
 =&\left(1-e^{-\frac{\epsilon}{P}} \right)^q\prod_{p=2}^{M+1}\left(1- {\rm P}( E_{jp})\right) ^q.\label{eqx3}
 \end{align}
 Similarly, ${\rm P}(  E_{m}^{n})$ can be obtained as follows:
   \begin{align}\label{eqx4}
{\rm P}(  E_{m}^{n}) =& \prod_{p=m}^{n}\left(1- {\rm P}( E_{jp})\right) .
 \end{align}
 By substituting \eqref{eqx2},  \eqref{eqx3} and  \eqref{eqx4} into \eqref{xdfd}, a  closed-form expression of $\mathcal{E}\{Y\}$  can be obtained. $\mathcal{E}\{Y^2\}$ can be evaluated similarly to $\mathcal{E}\{Y\}$, and   hence the lemma is proved.

\bibliographystyle{IEEEtran}
\bibliography{IEEEfull,trasfer}
  \end{document}